\RequirePackage[l2tabu, orthodox]{nag}
\documentclass[a4paper,11pt]{article}
\pdfoutput=1

\usepackage{fixltx2e}
\usepackage[T1]{fontenc}
\usepackage[utf8]{inputenc}
\usepackage{lmodern}
\usepackage{microtype}
\usepackage[margin=1in]{geometry}
\usepackage[UKenglish]{babel}
\usepackage{latexsym}
\usepackage{mathrsfs}
\usepackage{amsmath}
\usepackage{amssymb}
\usepackage{amsthm}
\usepackage[all,warning]{onlyamsmath}
\usepackage{enumitem}
\usepackage{xparse}
\usepackage{mathtools}
\usepackage{graphicx}
\usepackage[hidelinks]{hyperref}

\theoremstyle{plain}
\newtheorem{theorem}{Theorem}
\newtheorem{lemma}[theorem]{Lemma}

\theoremstyle{definition}
\newtheorem{definition}[theorem]{Definition}
\newtheorem{example}[theorem]{Example}

\setlist[itemize]{label=--}
\setlist[enumerate]{label=(\arabic*),labelindent=\parindent,leftmargin=*}

\DeclarePairedDelimiter\braces{\{}{\}}
\NewDocumentCommand\set{O{}mg}{\ensuremath{\braces[#1]{#2\IfNoValueTF{#3}{}{\,:\,#3}}}}
\DeclarePairedDelimiter{\ceil}{\lceil}{\rceil}

\DeclareMathOperator{\dist}{dist}

\newcommand{\Pow}{\mathcal{P}} 
\newcommand{\N}{\mathbb{N}} 
\newcommand{\Np}{\mathbb{N_+}} 

\newcommand{\local}{\mathsf{LOCAL}}
\newcommand{\ThueMorse}{\mathsf{ThueMorse}}
\newcommand{\yes}{\mathsf{yes}}
\newcommand{\no}{\mathsf{no}}

\newcommand{\A}{\mathcal{A}} 

\renewcommand{\_}{\underline{\phantom{0}}} 

\hypersetup{
  pdftitle={Constant Space and Non-Constant Time in Distributed Computing},
  pdfauthor={Tuomo Lempiäinen, Jukka Suomela}
}

\begin{document}

\title{\textbf{Constant Space and Non-Constant Time\\in Distributed Computing}}
\author{Tuomo Lempiäinen and Jukka Suomela\bigskip\\
  \small{Department of Computer Science, Aalto University, Finland}}
\date{}
\maketitle

\begin{abstract}
  While the relationship of time and space is an established topic in traditional centralised complexity theory, this is not the case in distributed computing. We aim to remedy this by studying the time and space complexity of algorithms in a weak message-passing model of distributed computing. While a constant number of communication rounds implies a constant number of states visited during the execution, the other direction is not clear at all. We consider several graph families and show that indeed, there exist non-trivial graph problems that are solvable by constant-space algorithms but that require a non-constant running time. This provides us with a new complexity class for distributed computing and raises interesting questions about the existence of further combinations of time and space complexity.
\end{abstract}

\section{Introduction}

This work studies the relationship between time and space complexity in distributed computing. While the relationships between various time complexity classes have been studied a lot in many of the established works of the field~\cite{goos13local-approximation,kuhn16local,linial92locality,naor95what} and also very recently~\cite{brandt2016lower,brandt2017lcl,chang2016exponential,chang2017time,ghaffari2017complexity}, space complexity has not received that much attention. We aim to remedy that by establishing the existence of a new complexity class: constant-space with non-constant running time.

\subsection{Distributed computing}

We consider a model of computation where each node of a graph is a computational unit. The same graph is a communication network and an input to the algorithm. Adjacent nodes communicate with each other in synchronous rounds, and eventually each node outputs its own part of the output. All the nodes run the same deterministic algorithm. In our model, the nodes are anonymous, that is, they do not have access to unique identifiers, and furthermore, they cannot distinguish between their neighbours---they broadcast the same message to everyone and they receive the messages in a set.

We define running time as the number of communication rounds until all the nodes have halted, while space usage is defined as the number of bits per node needed to represent all the states of the algorithm. The complexity measures are considered as a function of $n$, the number of nodes in the graph. The amount of local computation is not limited in any way, nor is the size of messages. More formally, distributed algorithms can be defined as state machines; see Section~\ref{sec:model} for the definitions.

\subsection{Contributions}

We prove the existence of graph problems, and consequently algorithms, that exhibit a constant space usage but non-constant running time. Our first theorem, proved in Section~\ref{sec:degree3} shows that when considering a graph family containing graphs of maximum degree 3, there exists a graph problem solvable by an algorithm with a constant number of states, while the time complexity of the problem is in $\Omega(\log n)$.

In our main theorem presented in Section~\ref{sec:degree2}, we turn our attention to the case of the class of graphs of maximum degree 2. While more non-trivial, it turns out that also in this case we can establish the existence of a graph problem solvable in constant space but requiring a linear number of communication rounds. Our result makes use of the properties of the Thue--Morse sequence, which we introduce in Section~\ref{sec:thuemorse}.

We emphasise that while it is straightforward to come up with such problems if we have a promise that the input graph is for example a path, we do not need to make such assumptions. It is also essential that we require the algorithms to halt in all nodes of any input graph. Again, this kind of results would be easy to achieve, if nodes were allowed to continue running indefinitely on some input graphs.

\subsection{Motivation and related work}

Traditionally, the focus of distributed computing has been on the time complexity of algorithms, and occasionally on message complexity, but space complexity has not been studied much. We aim to change this by introducing space complexity, constant space in particular, as a new dimension in the classification of computational problems.

While space complexity is especially interesting from the theory perspective, it can be argued that constant space complexity is a reasonable assumption also from the practical point of view. Networked devices often need to be able to operate in a network of any size, while it is not necessarily easy to expand their memory resources afterwards. On the other hand, also nature provides us with plenty of phenomena that exhibit distributed behaviour. Natural organisms are usually of constant size, independent of the size of the swarm or flock they are part of. Therefore, the study of constant space may result in new avenues for applying distributed computing to advance the understanding of nature.

The line of research on anonymous models of distributed computing was initiated by Angluin~\cite{angluin80local}, who introduced the well-known port numbering model. Our model of computation can be seen as a further restriction of the port-numbering model, with port numbers stripped out. While port numbers are a natural assumption in wired networks, the weaker variant makes more sense when applying distributed computing to a wireless setting. While our model has not been studied that much in prior work, a so-called \emph{beeping model}~\cite{afek11beeping,cornejo10deploying} is essentially similar. In the hierarchy of seven models defined by Hella et al.~\cite{hella15weak-models}, our model is the weakest one; they call it the $\mathsf{SB}$ model. On the other hand, the case where nodes receive messages in a multiset instead of a set is discussed more often in the literature~\cite{hella15weak-models}.

From the constant-space point of view, our setting bears similarities to the field of cellular automata~\cite{gardner1970fantastic,neumann1966theory,wolfram2002new}. In a cellular automaton, each cell can be in one of a constant number of states, and each cell updates its state synchronously using the same rule. However, in the case of cellular automata, one is usually interested in the kind of patterns an automata converges in, while we require each node to eventually stop and produce an output.

On the side of distributed computing, Emek and Wattenhofer~\cite{emek13stone} have considered a model where the network consists of finite state machines---hence making the space complexity constant. However, their model is asynchronous and randomised, while we study a fully synchronous and deterministic setting.

One of the main ingredients of our work, the Thue--Morse sequence, was used previously by Kuusisto~\cite{kuusisto14infinite}, who proved that there exists a distributed algorithm that always halts in the class of graphs of maximum degree two but features a non-constant running time. However, his algorithm has also a non-constant space complexity---this is where our work provides a significant improvement.

\section{Preliminaries}\label{sec:prelim}

In this section we define our model of computation and introduce notions needed later in the proofs.

\subsection{Model of computation}\label{sec:model}

Our model of computation is a weaker variant of the standard $\local$ model, with the following restrictions:
\begin{itemize}
  \item Nodes are anonymous, that is, they do not have unique identifiers.
  \item Nodes broadcast the same message to all their neighbours.
  \item Nodes receive the messages in a set, that is, they do not know which neighbour sent which message or how many identical messages they received.
\end{itemize}

In the following, we define distributed algorithms as state machines. Given an input graph, each node of the graph is equipped with an identical state machine. Machines in adjacent nodes can communicate with each other. In this work, we study only deterministic state machines and synchronous communication. The graph is always assumed to be simple, finite, connected and undirected, unless stated otherwise.

At the beginning, each state machine is only aware of its own local input (taken from some fixed finite set) and the degree of the node on which it sits. Then, computation is executed in synchronous rounds. In each round, each machine
\begin{enumerate}
  \item broadcasts a message to its neighbours,
  \item receives a set of messages from its neighbours,
  \item moves to a new state based on the received messages and its previous state.
\end{enumerate}
Each machine is required to eventually reach one of special halting states and stop execution. The local output is then the state of the node at the time of halting.

Note that while our model of computation is rather weak, it only makes our results stronger by limiting the capabilities of algorithms. Unique identifiers or unrestricted local inputs would not make much sense in the constant-space setting. Crucially, we require nodes to always halt in any input graph---nodes being allowed to run indefinitely or having the ability to continue passing messages after announcing an output would make engineering constant-space non-constant-time problems quite straightforward.

Next, we give a more formal definition of the model of computation used in this work by defining algorithms and graph problems.

\subsection{Notation and terminology}

For $k \in \N$, we denote by $[k]$ the set $\set{1,2,\ldots,k}$. Given a graph~$G=(V,E)$, the set of neighbours of a node $v \in V$ is denoted by $N(v) = \set{u\in V}{\set{v,u} \in E}$. For $r \in \N$, the \emph{radius-$r$ neighbourhood} of a node~$v\in V$ is $\set{u\in V}{\dist(u,v) \le r}$, that is, the set of nodes~$u$ such that there is a path of length at most $r$ between $v$ and $u$. We call any induced subgraph of $G$ that contains node~$v$ simply a \emph{neighbourhood} of $v$.

We will also work with strings of letters. Given a finite alphabet~$\Sigma$, we denote elements of the alphabet by lowercase symbols such as $x \in \Sigma$ or $y \in \Sigma$. On the other hand, \emph{words} (finite sequences of letters) are denoted by uppercase symbols, for example $X = x_1x_2\ldots x_i \in \Sigma^*$. Given words $X = x_1x_2\ldots x_i$ and $Y = y_1y_2\ldots y_j$, we write their concatenation simply $XY = x_1x_2\ldots x_i y_1y_2\ldots y_j$. A word~$X$ is a \emph{subword} of $Y$ if $Y = Y_1XY_2$ for some (possibly empty) words $Y_1$ and $Y_2$. We identify words of length 1 with the letter they consist of. For any letter $x$ and $i \in \N$, $x^i$ denotes the word consisting of $i$ consecutive letters $x$. We say that a word $X$ is \emph{of the form} $x+$ if $X = x^i$ for some $i \in \Np$.

\subsection{Algorithms as state machines}

Let $G = (V,E)$ be a graph. An \emph{input} for $G$ is a function $f\colon V\to I$, where $I$ is a finite set. For each node $v \in V$, we call $f(v) \in I$ the \emph{local input} of $v$.

A \emph{distributed state machine} is a tuple $\A = (S,H,\sigma_0,M,\mu,\sigma)$, where
\begin{itemize}
  \item $S$ is a set of states,
  \item $H \subseteq S$ is a finite set of halting states,
  \item $\sigma_0 \colon \N \times I \to S$ is an initialisation function,
  \item $M$ is a set of possible messages,
  \item $\mu \colon S \to M$ is a function that constructs the outgoing messages,
  \item $\sigma \colon S \times \Pow(M)$ is a function that defines the state transitions, so that $\sigma(h,\mathcal{M}) = h$ for each $h \in H$ and $\mathcal{M} \in \Pow(M)$.
\end{itemize}

Given a graph $G$, and input $f$ for $G$ and a distributed state machine~$\A$, the \emph{execution} of $\A$ on $(G,f)$ is defined as follows. The state of the system in round $r \in \N$ is a function $x_r \colon V \to S$, where $x_r(v)$ is the state of node $v$ in round~$r$. To begin the execution, set $x_0(v) = \sigma_o(\deg(v),f(v))$ for each node~$v \in V$. Then, let $A_{r+1}(v) = \set{\mu(x_r(u))}{u \in N(v)}$ denote the set of messages received by node~$v$ in round~$r+1$. Now the new state of each node~$v\in V$ is defined by setting $x_{r+1}(v) = \sigma(x_r(v),A_{r+1}(v))$.

The \emph{running time} of $\A$ on $(G,f)$ is the smallest $t\in \N$ for which $x_t(v) \in H$ holds for all $v \in V$. The output of $\A$ on $(G,f)$ is then $x_t\colon V\to H$, where $t$ is the running time, and for each $v\in V$, the \emph{local output} of $v$ is $x_t(v)$. The \emph{space usage} of $\A$ on $(G,f)$ is defined as
\[
  \ceil{\log |\set{s \in S}{s = x_r(v)\text{ for some } r\in\N, v\in V}|},
\]that is, the number if bits needed to encode all the states that the state machine visits in at least one node of $G$ during the execution. In case the execution does not halt, the running time can be defined to be $\infty$, and if the number of visited states grows arbitrarily large, we take the space usage to also be $\infty$.

From now on, we will use the terms \emph{algorithm~$\A$} and \emph{distributed state machine~$\A$} interchangeably, implying that each algorithm can be defined formally as a state machine.

\subsection{Graph problems}

The computational problems that we consider are graph problems with local input---that is, the problem instance is identical to the communication network, but possibly with an additional input value given to each node.

More formally, let $I$ and $O$ be finite sets. A \emph{graph problem} is a mapping $\Pi_{I,O}$ that maps each graph~$G=(V,E)$ and input $f \colon V \to I$ to a set $\Pi_{I,O}(G,f)$ of valid solutions. Each solution~$S$ is a function $S\colon V \to O$. In case of \emph{decision graph problems}, we set $O = \set{\yes,\no}$.

If $\Pi_{I,O}$ is a graph problem, $T,U\colon \N \to \N$ are functions and $\A$ a distributed state machine, we say that \emph{$\A$ solves $\Pi_{I,O}$ in time $T$ and in space $U$} if for each graph $G=(V,E)$ and each input $f\colon V\to I$ we have that the running time of $\A$ on $(G,f)$ is at most $T(|V|)$, the space usage of $\A$ on $(G,f)$ is at most $U(|V|)$ and the output of $\A$ on $(G,f)$ is in the set $\Pi_{I,O}(G,f)$. In that case, we also say that the \emph{time complexity} of algorithm~$\A$ is $T$ and the \emph{space complexity} of $\A$ is~$U$.

\subsection{Thue--Morse sequence}\label{sec:thuemorse}

In this section we present a concept that will be central in the proof of our main result in Section~\ref{sec:degree2}. The Thue--Morse sequence is the infinite binary sequence defined recursively as follows:
\begin{definition}
  The \emph{Thue--Morse} sequence is the sequence $(t_i)$ satisfying $t_0 = 0$, and for each $i \in \N$, $t_{2i} = t_i$ and $t_{2i+1} = 1-t_i$.
\end{definition}
Thus, the beginning of the Thue--Morse sequence is
\[
  0110 1001 1001 0110 1001 \ldots
\]
For our purposes, the following two recursive definitions will be very useful.
\begin{definition}
  Let $T_0 = 0$. For each $i \in \Np$, let $T_i = T_{i-1}\mathcal{C}(T_{i-1})$, where $\mathcal{C}$ denotes the Boolean complement.
  \label{def:tm-compl}
\end{definition}
Note that for each $i \in \N$, the word $T_i$ the prefix of length $2^i$ of the Thue--Morse sequence.
\begin{definition}
  Let $T'_0 = 0$. For each $i \in \Np$, let $T'_i$ be obtained from $T'_{i-1}$ by substituting each occurrence of 0 with 01 and each occurrence of 1 with 10.
  \label{def:tm-rec}
\end{definition}
Again, we have $T'_0 = 0$, $T'_1 = 01$, $T'_2 = 0110$, $T'_3 = 01101001$ and so on. A straightforward induction shows that the above two definitions are equivalent: $T_i = T'_i$ for each $i \in \N$. We call $T_i$ the \emph{Thue--Morse word of length $2^i$}.

The Thue--Morse sequence contains lots of squares, that is, subwords of the form $XX$, where $X \in \set{0,1}^*$. Interestingly, it does not contain any cubes---subwords of the form $XXX$. Note also that for each $i \in \Np$, $T_{2i}$ is a palindrome.

\section{Warm-up: graphs of maximum degree 3}\label{sec:degree3}

In this section we present a graph problem that exhibits the constant space and non-constant time complexity, in case we do not restrict ourselves to paths and cycles. The proof is quite straightforward, which emphasises the fact that the degree-2 case considered in Section~\ref{sec:degree2} is the most interesting one. We emphasise that in the following theorem, we do not need to make any additional assumptions about the graph; the described algorithm halts in all finite input graphs.

\begin{theorem}
  There exist a graph problem with constant space complexity and $\Theta(\log n)$ time complexity in the class of graphs of maximum degree at most 3.
  \label{thm:degree3}
\end{theorem}

\begin{proof}
  Consider the following transformation: given a \emph{directed} graph $G' = (V',E')$, replace each edge $e=(u,v)$ by the gadget represented in Figure~\ref{fig:gadget}. The end result is an undirected graph $G = (V,E)$, where the gadgets encode the edge directions of the original graph. We say that an undirected graph~$G$ is \emph{good} if it is obtained from some directed binary pseudotree~$G'$ (that is, a connected graph where each node has outdegree at most one and indegree either 0 or 2) by the above transformation.
  \begin{figure}
    \centering
    \includegraphics[page=1]{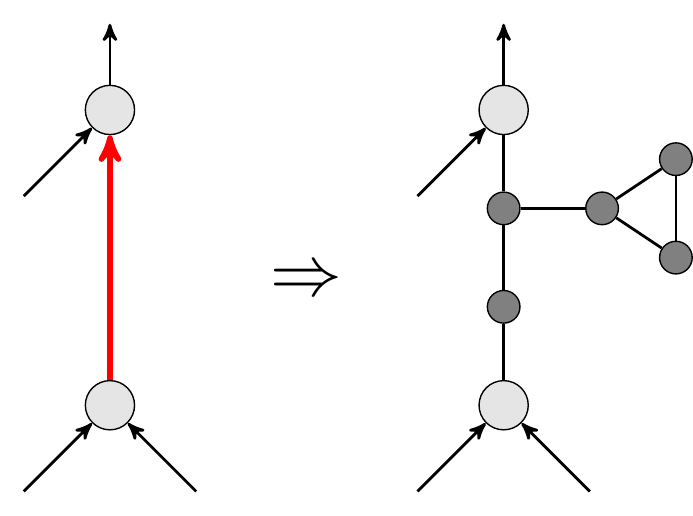}
    \caption{The red edge is replaced by the dark grey gadget consisting of five nodes. Similar encoding is performed to each edge of the original directed graph.}
    \label{fig:gadget}
  \end{figure}

  Consider the following graph problem: if the input graph is good, each node has to output the parity of its distance to the nearest leaf node; if not, each node is allowed to output anything. In any case, all the nodes have to halt.

  Define an algorithm~$\A$ as follows. First, $\A$ checks locally whether the input graph~$G$ is good: Each node can detect in at most four communication rounds, whether it is contained in a valid gadget. Then, each node~$v$ corresponding to a node $v'$ from the original directed graph can check that in the orientation obtained from the gadgets, the node either has two incoming edges and one outgoing edge, only one outgoing edge (a leaf node), or only two incoming edges (the root node in case of a tree). If either of the checks fails, the node broadcasts an instruction to halt to its neighbours and halts. Otherwise, it follows that the original graph~$G'$ is indeed a directed binary pseudotree with each edge directed away from the leaf nodes.

  After verifying the goodness of the instance, algorithm~$\A$ starts to count the distances to leaf nodes. In the first round, each leaf node broadcasts message 0 to its neighbours, halts and outputs 0. When a node that has not yet halted receives message $i$, where $i \in \set{0,1}$, it broadcasts $i+1 \mod 2$ to its neighbours, halts and outputs $i+1 \mod 2$. Because the input graph is assumed to be finite and connected, each node eventually halts.

  Since verifying the goodness takes only a constant number of rounds, and in the counting phase we count only up to 2, algorithm~$\A$ needs only a constant number of states. Furthermore, as in any binary pseudotree the distance from any node to the closest leaf node is at most $O(\log n)$, each node halts after at most $O(\log n)$ rounds. On the other hand, in a balanced binary tree, the distance from the root to leaf nodes is $\Omega(\log n)$, and thus $\Omega(\log n)$ rounds are needed until all the nodes can output the parity of their distance to the closest leaf. This completes the proof.
\end{proof}

\section{Graphs of maximum degree 2}\label{sec:degree2}

We are now ready to state the main theorem of this work.

\begin{theorem}
  There exist a decision graph problem with constant space complexity and $\Theta(n)$ time complexity in the class of graphs of maximum degree at most 2.
  \label{thm:degree2}
\end{theorem}

To define our graph problem, we will make use of the Thue--Morse sequence via the following definition---compare this to Definition~\ref{def:tm-rec} earlier.
\begin{definition}[Valid words]
  Define a set of words over $\set{0,1,\_}$ recursively as follows:
  \begin{enumerate}
    \item $\_0\_$ is \emph{valid},
    \item if $X$ is valid and $Y$ is obtained from $X$ by applying the substitutions $0 \mapsto 0\_1\_1\_0$ and $1 \mapsto 1\_0\_0\_1$ to each occurrence of $0$ and $1$, then $Y$ is \emph{valid}.
  \end{enumerate}
\end{definition}

That is, valid words are obtained from the sequences $T_{2i}, i \in \N$, by inserting an underscore at the beginning, in between each symbol and at the end.

Now, we will define a decision graph problem that we call $\ThueMorse$ as follows.

\begin{definition}[$\ThueMorse$]
  The local inputs of the problem are taken from the set $I = \set{\alpha,\beta,\gamma} \times \set{0,1,\_}$ and each local output is either $\yes$ or $\no$. Given an input function $f \colon V \to I$, we write $f = (f_1,f_2)$. Now, an instance is a yes-instance of $\ThueMorse$ if and only if
  \begin{itemize}
    \item the graph is a path graph,
    \item the first parts $f_1(\cdot)$ of the local inputs define a consistent orientation,
    \item the second parts $f_2(\cdot)$ of the local input define a valid word.
  \end{itemize}
\end{definition}

\subsection{Definition of the algorithm}

Next, we will give an algorithm that is able to solve the decision problem $\ThueMorse$ by using only constant space. The high-level idea is as follows: First, we check that the local neighbourhood of each node looks correct. Then, we repeatedly apply substitutions that roll the configuration of the path graph back to a shorter prefix of the Thue--Morse sequence, until we reach a trivial configuration and accept the input---or fail to apply the substitutions unambiguously and consequently reject the input.

\begin{lemma}
  There exists an algorithm~$\A$ that solves problem $\ThueMorse$ in space $O(1)$ and time~$O(n)$.
  \label{lem:deg2alg}
\end{lemma}

First, let us introduce some terminology and notation. In an instance $(G,f)$ of $\ThueMorse$ where $G$ is either a path or a cycle graph, the sequence of symbols defined by $f_2$ in a given node neighbourhood is called the \emph{input word}---note that contrary to usual words, the input word is unoriented. Sometimes we identify the node neighbourhood with the corresponding input word; the meaning should be clear from the context. During the execution of the algorithm, each node $v$ has a \emph{current symbol}~$c(v)$ as part of its state. The vertical bar $|$ shall denote the end of the path graph. We will make use of it by assuming that each degree-$1$ node~$v$ sees a ``virtual'' neighbour~$u$ with the symbol~$|$, that is, $f_1(u) = f_2(u) = |$ and $c(u)$ is always equal to $|$. The underscore symbols~$\_$ will be called \emph{separators}. Denote the alphabet which contains the possible current symbols by $\Sigma = \set{0,1,\_,|}$.

We define the algorithm~$\A$ in three parts, which we denote by I, II and III. In each part, any node can \emph{abort}, which means that it sends a special abort message to its neighbours and then moves to state $\no$ and thus halts. If a node receives the abort message at any time, it aborts---that is, it passes the message on to all its neighbours, and then moves to state $\no$ and halts. At initialisation, each node~$v$ sets its symbol~$c(v)$ to be equal to $f_2(v)$, the second part of the local input.

In {\bfseries part I}, each node~$v$ verifies its degree and the orientation: if $\deg(v) \in \set{1,2}$ and three different symbols from $\set{\alpha,\beta,\gamma,|}$ can be found in the local inputs within the radius-$1$ neighbourhood of $v$ (that is, we have $|\set{f_1(u)}{u\in N(v) \cup \set{v}}| = 3$), continue; otherwise, abort. Recall that the graph is assumed to be finite and connected. If none of the nodes aborts, it follows that the graph is either a path or a cycle and that the local inputs given by $f_1$ define a word of the form $\ldots \alpha\beta\gamma\alpha\beta\gamma\alpha\beta\gamma\alpha \ldots$.

In {\bfseries part II}, each node~$v$ verifies the input word in its radius-$1$ neighbourhood: if the neighbourhood is in $\set{|\_0,\;0\_|,\;0\_0,\;1\_1,\;0\_1,\;1\_0,\;\_0\_,\;\_1\_}$, continue; otherwise, abort. Note that here we do not care about the orientation of the word, and hence the set of symbols received from the neighbours is enough for this step. If none of the nodes aborts, the input word is \emph{locally valid}: every other symbol is the separator $\_$ and every other symbol is either 0 or 1.

Then, we proceed to {\bfseries part III}, which contains the most interesting steps of the algorithm. Now we will make use of the orientation given as part of the input. Define $\alpha < \beta$, $\beta < \gamma$ and $\gamma < \alpha$. For each node~$v$, we say that neighbour~$u$ of $v$ is the \emph{left neighbour} of $v$ if $f_1(u) < f_1(v)$ and the \emph{right neighbour} of $v$ if $f_1(u) > f_1(v)$. Thus, each node can essentially send a different message to each of its two neighbours: the orientation symbol of the recipient indicates which part of the message is intended for which recipient. From now on, we will also assume that each node attaches its own orientation symbol as part of every message sent, so that nodes can distinguish between incoming messages.

Part III will consist of several \emph{phases}. In each phase, we first gather information from the neighbourhood in two \emph{buffers} and then try to apply a substitution to the word obtained in the buffers. More precisely, each node~$v$ has two buffers, the left buffer $L(v)$ and the right buffer $R(v)$ as part of its state. The buffers are used to store a \emph{compressed} version of the input word in the neighbourhood: a sequence of consecutive symbols 0 or 1 is represented by a single 0 or 1, respectively.

Let us define some notation. Let $l,r \colon \Sigma^* \times \Sigma \to \Sigma^*$ be as follows. If $X = Ay$ for some word $A$, set $r(X,y) = X$. Otherwise, set $r(X,y) = Xy$. Function~$l$ is defined analogously: if $X = yA$ for some word $A$, set $l(X,y) = X$; otherwise, set $l(X,y) = yX$. In other words, functions $l$ and $r$ append a new symbol either to the beginning or to the end of a word, respectively---but only if the new symbol is different from the current first or last symbol of the word, respectively.

Now, in each phase, each node~$v$ first \emph{fills its buffers} as follows. To initialise, node~$v$ broadcasts message $c(v)$ (its own current symbol) to its neighbours. Then, node~$v$ repeats the following. It sets the left buffer $L(v)$ equal to the message it receives from its left neighbour and the right buffer $R(v)$ equal to the message it receives from the right neighbour. After that, node~$v$ sends $r(L(v),c(v))$ to the right and $l(R(v),c(v))$ to the left. These steps are repeated until both buffers of~$v$ contain either 8 instances of the separator~$\_$ or an end-of-the-path marker~$|$. When that happens, node~$v$ is finished with filling its buffers.

Next, node~$v$ combines its buffers to construct a \emph{compressed view} of the input word in its neighbourhood. To that end, define $C\colon \Sigma^* \times \Sigma \times \Sigma^* \to \Sigma^*$ and $p\colon \Sigma^* \times \Sigma \times \Sigma^* \to \N$ as follows. If $X = Ay$ and $Z = yB$ for some $A,B \in \Sigma^*$, set $C(X,y,Z) = AyB$ and $p(X,y,Z) = |A|$. Otherwise, if $X = Ay$ for some $A \in \Sigma^*$, set $C(X,y,Z) = AyZ$ and $p(X,y,Z) = |A|$, and if $Z = yB$ for some $B \in \Sigma^*$, set $C(X,y,Z) = XyB$ and $p(X,y,Z) = |X|$. Else, set $C(X,y,Z) = XyZ$ and $p(X,y,Z) = |X|$. Now, the compressed view of node~$v$ is $V(v) = C(L(v),c(v),R(v))$, and $c(v)$ is at position $q(v) = p(L(v),c(v),R(v))+1$ in $V(v)$. In other words, the left buffer, the current symbol and the right buffer are concatenated in a way that removes successive repetitions of the same symbol.

Finally, node~$v$ does subword matching on the view~$V(v)$. If $V(v)$ is equal to $|\_0\_|$ or $|\_0\_1\_1\_0\_|$, node~$v$ instructs other nodes to accept, moves to the $\yes$ state and halts. Otherwise, node~$v$ searches $V(v)$ for the subword $\_0\_1\_1\_0\_1\_0\_0\_1\_$ in all possible positions. Given a match, let the $i$th symbol of the subword be aligned with the $q(v)$th symbol of $V(v)$. If $i \in \set{1,9,17}$, set $c' = \_$. Otherwise, if $i\le 8$, set $c' = 0$, and if $i \ge 10$, set $c' = 1$. After that, node~$v$ performs the same procedure with the reversed subword $\_1\_0\_0\_1\_0\_1\_1\_0\_$ ---but now, if $i\le 8$, set $c' = 1$, and if $i \ge 10$, set $c' = 0$. If no matches could be found in $V(v)$, node~$v$ aborts. If several matches were found and they resulted in different values for $c'$, node~$v$ aborts. Otherwise, node~$v$ updates its current symbol $c(v)$ to be the unambiguous value $c'$. This concludes the phase; if not aborted, on the next round, a new phase starts from the beginning. See Example~\ref{ex:part3} for illustrations of the matching and substitution steps in a few cases.

\begin{example}
  Consider the execution of part III of algorithm~$\A$ in the following instances.
  \begin{enumerate}
    \item Path graph, yes-instance:
      \begin{align*}
        |\_0\_1\_1\_0\_1\_0\_0\_1&\_1\_0\_0\_1\_0\_1\_1\_0\_| \\
        &\Downarrow \quad(\text{unambiguous substitutions})\\
        |\_0000000\_1111111&\_1111111\_0000000\_| \\
        &\Downarrow \quad(\text{$V(v) = |\_0\_1\_1\_0\_|$})\\
        &\hspace{-0.7em}\text{accept}
      \end{align*}
    \item Path graph, no-instance:
      \begin{align*}
        |\_0\_1\_1\_0\_1\_0\_0&\_1\_1\_0\_1\_0\_0\_1\_| \\
        &\Downarrow \\
        |\_0000000\_11111&11\_ \ldots \\
        \ldots \_0&000000\_1111111\_| \\
        &\Downarrow \quad(\text{ambiguous substitutions})\\
        &\hspace{-0.6em}\text{abort}
      \end{align*}
    \item Cycle graph (the ends marked with $\ldots$ are connected circularly):
      \begin{align*}
        \ldots \_0\_1\_1\_0\_1\_0\_0\_1&\_1\_0\_0\_1\_0\_1\_1\_0\_ \ldots \\
        &\Downarrow \quad(\text{unambiguous substitutions})\\
        \ldots \_0000000\_1111111&\_1111111\_0000000\_ \ldots \\
        &\Downarrow \quad(\text{no matches})\\
        &\hspace{-0.6em}\text{abort}
      \end{align*}
  \end{enumerate}
  \label{ex:part3}
\end{example}

\subsection{Proof of correctness}

In this section we show that algorithm~$\A$ executes correctly and always halts with the desired output. Since parts I and II are quite trivial, we will start with part III of the algorithm.

We say that a word~$X = x_1 x_2 \ldots x_i$ is a \emph{compressed version} of word~$Y = y_1 y_2 \ldots y_j$ if $x_i \ne x_{i+1}$ for all $i \in \set{1,2,\ldots,i-1}$ and there exist a surjective \emph{compression mapping} $f \colon [j] \to [i]$ such that $f(1) = 1$, $f(k) \in \set{f(k-1),f(k-1)+1}$ for all $k \in \set{2,3,\ldots,j}$ and $y_k = x_{f(k)}$ for all $k \in [j]$.

\begin{lemma}
  After any node~$v$ has finished collecting the buffers, $V(v)$ is the compressed version of the actual input word in the neighbourhood of $v$.
  \label{lem:compress}
\end{lemma}

\begin{proof}
  We use induction on the number of rounds~$t$ after starting the phase. After the first round of the phase, each node~$v$ has received $L(v) = c(u)$ from its left neighbour~$u$ and $R(v) = c(w)$ from its right neighbour~$w$. Hence $L(v)$ and $R(v)$ are compressed versions of the left and right 1-neighbourhoods of $v$, respectively.

  Assume then that $L(u)$ and $R(w)$ are compressed versions of the left and right $(t-1)$-neighbourhoods of $u$ and $w$, respectively, and let $f_l$ and $f_r$ be the corresponding mappings. Consider now the definition of the algorithm. The definition of the mapping $r$ implies that what $v$ receives from the left in round~$t$ of the phase is $L(u)$---extended by $c(u)$ if and only if $c(u)$ differs from the last symbol of $L(u)$. If it does differ, we extend $f_l$ by defining $f_l(t) = |L(u)|+1$, otherwise $f_l(t) = |L(u)|$. Hence the new value of the buffer $L(v)$ is a compressed version of the left $t$-neighbourhood of $v$. The case of $R(v)$ is handled analogously.

  Suppose that $v$ has finished collecting the buffers. Now $L(v)$ and $R(v)$ are the compressed versions of the left and right $k$-neighbourhoods of $v$ for some $k$. Then, it follows from the definition of the function $C$ that an appropriate compression mapping can be formed and $V(v) = C(L(v),c(v),R(v))$ is the compressed version of the $k$-neighbourhood of $v$.
\end{proof}

We call the sequence of current symbols $c(u)$ in the graph a \emph{configuration} (in the case of a cycle graph, the sequence is infinite in both directions). A maximal sequence of adjacent nodes such that each node~$v$ has the same current symbol~$c(v)$ is called a \emph{block}. We sometimes identify a block with the subword consisting of the current symbols of the block nodes.

In the next two lemmas, we show how updating the current symbol locally in nodes results in global substitutions in the configuration.

\begin{lemma}
  Assume that in the current configuration, each maximal subword of the form $0+$ or $1+$ is of length $\ell$ and each maximal subword of the form $\_+$ is of length 1. If the algorithm is executed for one phase and no node aborts, in the resulting configuration the lengths are $4\ell+3$ and 1, respectively.
  More precisely, the execution of one phase always results in substitutions of the following kinds:
  \begin{align}
    \_0^\ell\_1^\ell\_1^\ell\_0^\ell\_1^\ell\_0^\ell\_0^\ell\_1^\ell\_ &\mapsto \_0^{4\ell+3}\_1^{4\ell+3}\_,\\
    \_1^\ell\_0^\ell\_0^\ell\_1^\ell\_0^\ell\_1^\ell\_1^\ell\_0^\ell\_ &\mapsto \_1^{4\ell+3}\_0^{4\ell+3}\_.
  \end{align}
  \label{lem:synchronous}
\end{lemma}

\begin{proof}
  Since no node aborts, each node~$v$ is able to find an unambiguous new value for $c(v)$. Consider an arbitrary node~$u$. Suppose that the pattern $P_1 = \_0\_1\_1\_0\_1\_0\_0\_1\_$ matches $V(u)$ so that the $i$th symbol of the pattern is aligned with the $q(v)$th symbol of $V(u)$. Now it follows from Lemma~\ref{lem:compress} that $P_1$ is the compressed version of an actual neighbourhood~$N$ of $u$. Due to the assumption, for each $j$ such that the $j$th symbol of $P_1$ is $0$ or $1$, in the neighbourhood there are exactly $\ell$ consecutive symbols $0$ or $1$, respectively, that are mapped to the $j$th symbol of $P_1$ by the compression mapping~$f$.

  Let us consider the case $i = 6$ as an example. Now $c(u) = 1$. As the buffers are gathered until each of them contain eight separators $\_$, the next 5 blocks to the left from the block of $u$, as well as the next 11 blocks to the right, gather views that are compressed versions of a neighbourhood containing~$N$. Hence $P_1$ matches also their views. For example, for all nodes $w$ in the block two steps left from the block of $u$, $P_1$ matches $V(w)$ at position $i-2 = 4$. If follows from the definition of the algorithm that the new value for $c(u)$, as well as for $c(w)$, where $w$ is in the next 4 blocks to the left from $u$ or 2 blocks to the right from $u$, is 0. The new value for the node in the 5th block left from $u$ as well as 3rd block right from $u$ is $\_$. Thus, after the phase, node~$u$ will be part of a $0$-block of length $\ell+1+\ell+1+\ell+1+\ell = 4\ell+3$.

  The cases for all other values of $i$, as for as the matching the reverse pattern $P_2$, are analogous.
\end{proof}

We call a word $\_x_1^i \_x_2^i \_ \ldots \_x_p^i \_$ a \emph{padded Thue--Morse word of length $p$} if $x_1 x_2 \ldots x_p$ is a prefix of the Thue--Morse sequence.

\begin{lemma}
  Let $W$ be a subword of a configuration $C$ at the beginning of a phase. Let $k \ge 3$. If $W$ is a (complement of a) padded Thue--Morse word of length $2^k$ and the algorithm is executed for one phase on $C$ without aborting, the subword $W$ is transformed to a (complement of a) padded Thue--Morse word of length $2^{k-2}$.
  \label{lem:padded}
\end{lemma}

\begin{proof}
  We use induction on $k$. If $k = 3$, we have $W = \_0^i\_1^i\_1^i\_0^i\_1^i\_0^i\_0^i\_1^i\_$ or $W = \_1^i\_0^i\_0^i\_1^i\_0^i\_1^i\_1^i\_0^i\_$ for some $i$. Then Lemma~\ref{lem:synchronous} implies that $W$ is transformed to $\_0^j\_1^j\_$ or $\_1^j\_0^j\_$, respectively, where $j = 4i+3$. Hence the claim holds for $k = 3$.

  Suppose then that $k > 3$ and the claim holds for each subword that is a (complement of~a) padded Thue--Morse word of length $2^{k-1}$. Let $W$ be a padded Thue--Morse word of length~$2^k$. Now the definition of the Thue--Morse sequence implies that we can write $W = W_1\_W_2$, where $W_1\_$ is a padded Thue--Morse word of length $2^{k-1}$ and $\_W_2$ is a complement of a padded Thue--Morse word of length $2^{k-1}$. By the inductive hypothesis, $W_1\_$ is transformed to a padded Thue--Morse word of length $2^{k-3}$ and $\_W_2$ is transformed to a complement of a padded Thue--Morse word of length $2^{k-3}$. Now the definition of the Thue--Morse sequence again implies that $W = W_1\_W_2$ is transformed to a padded Thue--Morse word of length $2^{k-2}$. The case where $W$ is a complement of a padded Thue--Morse word is completely analogous.
\end{proof}

Now we are ready to aggregate our previous lemmas to establish that algorithm~$\A$ actually works correctly.

\begin{lemma}
  In a yes-instance, each node eventually halts and outputs $\yes$.
  \label{lem:accept}
\end{lemma}

\begin{proof}
  Let $(G,f)$ be a yes-instance of $\ThueMorse$. By definition, algorithm~$\A$ executes parts I and II successfully. Notice then that since the input word is valid, the configuration at the beginning of part III is actually a padded Thue--Morse word. If the input word is either $\_0\_$ or $\_0\_1\_1\_0$, the nodes move immediately to the $\yes$ state and halt. Otherwise, the input word is a padded Thue--Morse word of length at least $16$, and by iterating Lemma~\ref{lem:padded} we obtain a shorter padded Thue--Morse word after every phase. Note that the new current symbol will be unambiguous for each node on each phase. Eventually the configuration will match with $\_0\_1\_1\_0\_$, and the instance will get accepted.
\end{proof}

\begin{lemma}
  In a no-instance, each node eventually halts and outputs $\no$.
  \label{lem:reject}
\end{lemma}

\begin{proof}
  Let $(G,f)$ be a no-instance of $\ThueMorse$. By assumption, $G$ is finite and connected. If $G$ contains a node of degree higher than 2, algorithm~$\A$ aborts in part~I. Hence we can assume that $G$ has maximum degree at most two. It follows immediately from Lemma~\ref{lem:synchronous} that algorithm~$\A$ halts on $(G,f)$: since the size of blocks of $0$'s or $1$'s grows in each phase, we will eventually run out of nodes.

  Graph~$G$ is either a path or a cycle graph. If $\A$ rejects in part I or part II, we are done. Hence, we can assume that the input~$f$ defines a consistent orientation and each 1-neighbourhood is of the correct form---that is, every second symbol on the input word is a separator~$\_$. Now it follows from Lemma~\ref{lem:synchronous} that the computation proceeds synchronously: in part III, each node starts a new phase in the same round.

  Consider the case that $G$ is a cycle graph. Due to Lemma~\ref{lem:synchronous}, the number of blocks decreases after each phase. Eventually, if no node rejects, the execution reaches a configuration with $b \le 4$ blocks of $0$'s and $1$'s. Then, each node sees the same sequence of $b$ blocks repeating in its view. It follows that neither pattern $P_1$ or $P_2$ matches and the node rejects.

  Assume then that $G$ is a path graph. If there exists a $\yes$-instance $(G',f')$ with the same number of nodes as $G$ has, we proceed as follows. Suppose for a contradiction that $(G,f)$ gets accepted. Then there is a smallest $i$ such that before the $i$th phase, the configurations $C$ on $(G,f)$ and $C'$ on $(G',f')$, respectively, are different, but after the $i$th phase, they are identical, $C''$. It follows that $C''$ contains a subword of the form $\_0^i\_1^i\_$ for some $i$, such that $C$ and $C'$ differ on a position overlapping with the subword. But this is a contradiction, as it follows from Lemma~\ref{lem:synchronous} that $C$ and $C'$ cannot differ on such a position.

  Finally, consider the case where $(G,f)$ is a no-instance such that there does not exist a yes-instance with the same number of nodes. Suppose again for a contradiction that $(G,f)$ gets accepted. Let $(G',f')$ be the largest yes-instance no larger than $(G,f)$ and let $(G'',f'')$ be the yes-instance that is one step larger from $(G',f')$. Lemma~\ref{lem:padded} implies that the algorithm~$\A$ needs exactly one phase more on $(G'',f'')$ than on $(G',f')$. The number of phases on $(G,f)$ equals either the number of phases on $(G',f')$ or on $(G'',f'')$. But either case is a contradiction, since the instances have a different amount of blocks in the beginning, and Lemma~\ref{lem:synchronous} implies that the size of blocks grows at a fixed rate.
\end{proof}

The last thing left to do is analysing the complexity of the algorithm. This is taken care of in the following lemmas.

\begin{lemma}
  The space usage of algorithm~$\A$ is in $O(1)$.
  \label{lem:space}
\end{lemma}

\begin{proof}
  Parts I and II of algorithm~$\A$ clearly use only a constant number of states. Consider then part III. When gathering the buffers, consecutive blocks of symbols in the neighbourhood are represented by only one symbol, and only a constant amount of blocks are gathered (eight separator symbols $\_$ in each buffer). Hence a constant number of states is enough to represent the contents of the buffers. Furthermore, the buffers get erased after each phase has completed. Thus, algorithm~$\A$ can be implemented using a constant number of states, independent of the size of the input.
\end{proof}

\begin{lemma}
  The running time of algorithm~$\A$ is in $O(n)$.
  \label{lem:time}
\end{lemma}

\begin{proof}
  Executing each phase of part III of the algorithm can be done in $8i+8 = 8(i+1)$ rounds, where $i$ is the size of the blocks at the start of the phase. Recall that after each phase, the size of the blocks grows from $\ell$ to $4\ell+3$. Note also that $\frac{1}{2}\log n$ phases are enough to grow the block size so large that the algorithm has to either accept or reject. It follows that
  \[
    8(1+1) + 8(7+1) + 8(31+1) + \cdots + 8(2^{2(\frac{1}{2}\log n)-1)}) \le 8n
  \]
  is an upper bound for the running time.
\end{proof}

On the other hand, $\Omega(n)$ rounds are clearly necessary to solve $\ThueMorse$: node at the end of the path has to receive information from the other end of the path to be able to verify that the instance is actually a yes-instance. This concludes the proof of Theorem~\ref{thm:degree2}.

\section{Conclusions}\label{sec:concl}

We identified a model of distributed computing and a class of graphs, where the question on the existence of constant-space, non-constant-time algorithms is interesting and non-trivial. The answer turned out to be positive. This opens up the way to study constant space further---we can ask, for example, what other time complexities besides $\Theta(\log n)$ and $\Theta(n)$ can possibly be found and in which graph classes.

\section*{Acknowledgements}

We have discussed this work with numerous people, including at least Juho Hirvonen, Janne~H. Korhonen, Antti Kuusisto, Joel Rybicki and Przemysław Uznański. Many thanks to all of them, as well as to others that we may have forgotten.

\def\UrlFont{\sf\footnotesize}
\bibliographystyle{plain}
\bibliography{constant-space}

\begin{thebibliography}{10}

\bibitem{afek11beeping}
Yehuda Afek, Noga Alon, Ziv Bar-Joseph, Alejandro Cornejo, Bernhard Haeupler,
  and Fabian Kuhn.
\newblock Beeping a maximal independent set.
\newblock In {\em Proc.\ 25th International Symposium on Distributed Computing
  (DISC 2011)}, volume 6950 of {\em Lecture Notes in Computer Science}, pages
  32--50. Springer, 2011.
\newblock
  \href{http://dx.doi.org/10.1007/978-3-642-24100-0_3}{\nolinkurl{doi:10.1007/978-3-642-24100-0_3}}.

\bibitem{angluin80local}
Dana Angluin.
\newblock Local and global properties in networks of processors.
\newblock In {\em Proc.\ 12th Annual ACM Symposium on Theory of Computing (STOC
  1980)}, pages 82--93. ACM Press, 1980.
\newblock
  \href{http://dx.doi.org/10.1145/800141.804655}{\nolinkurl{doi:10.1145/800141.804655}}.

\bibitem{brandt2016lower}
Sebastian Brandt, Orr Fischer, Juho Hirvonen, Barbara Keller, Tuomo
  Lempiäinen, Joel Rybicki, Jukka Suomela, and Jara Uitto.
\newblock A lower bound for the distributed {Lovász} local lemma.
\newblock In {\em Proc.\ 48th Annual ACM Symposium on Theory of Computing (STOC
  2016)}, pages 479--488. ACM Press, 2016.
\newblock
  \href{http://dx.doi.org/10.1145/2897518.2897570}{\nolinkurl{doi:10.1145/2897518.2897570}}.
  \href{http://arxiv.org/abs/1511.00900}{\nolinkurl{arXiv:1511.00900}}.

\bibitem{brandt2017lcl}
Sebastian Brandt, Juho Hirvonen, Janne~H. Korhonen, Tuomo Lempiäinen, Patric
  R.~J. {\"O}stergård, Christopher Purcell, Joel Rybicki, Jukka Suomela, and
  Przemysław Uznański.
\newblock {LCL} problems on grids.
\newblock In {\em Proc.\ 36th Annual ACM Symposium on Principles of Distributed
  Computing (PODC 2017)}, 2017.
\newblock To appear.
  \href{http://arxiv.org/abs/1702.05456}{\nolinkurl{arXiv:1702.05456}}.

\bibitem{chang2016exponential}
Yi-Jun Chang, Tsvi Kopelowitz, and Seth Pettie.
\newblock An exponential separation between randomized and deterministic
  complexity in the {LOCAL} model.
\newblock In {\em Proc.\ 57th Annual IEEE Symposium on Foundations of Computer
  Science (FOCS 2016)}, pages 615--624. IEEE Computer Society Press, 2016.
\newblock
  \href{http://dx.doi.org/10.1109/FOCS.2016.72}{\nolinkurl{doi:10.1109/FOCS.2016.72}}.
  \href{http://arxiv.org/abs/1602.08166}{\nolinkurl{arXiv:1602.08166}}.

\bibitem{chang2017time}
Yi-Jun Chang and Seth Pettie.
\newblock A time hierarchy theorem for the {LOCAL} model, 2017.
\newblock \href{http://arxiv.org/abs/1704.06297}{\nolinkurl{arXiv:1704.06297}}.

\bibitem{cornejo10deploying}
Alejandro Cornejo and Fabian Kuhn.
\newblock Deploying wireless networks with beeps.
\newblock In {\em Proc.\ 24th International Symposium on Distributed Computing
  (DISC 2010)}, volume 6343 of {\em Lecture Notes in Computer Science}, pages
  148--162. Springer, 2010.
\newblock
  \href{http://dx.doi.org/10.1007/978-3-642-15763-9_15}{\nolinkurl{doi:10.1007/978-3-642-15763-9_15}}.

\bibitem{emek13stone}
Yuval Emek and Roger Wattenhofer.
\newblock Stone age distributed computing.
\newblock In {\em Proc.\ 32nd Annual ACM Symposium on Principles of Distributed
  Computing (PODC 2013)}, pages 137--146. ACM Press, 2013.
\newblock
  \href{http://dx.doi.org/10.1145/2484239.2484244}{\nolinkurl{doi:10.1145/2484239.2484244}}.
  \href{http://arxiv.org/abs/1202.1186}{\nolinkurl{arXiv:1202.1186}}.

\bibitem{gardner1970fantastic}
Martin Gardner.
\newblock The fantastic combinations of {J}ohn {C}onway's new solitaire game
  ``life''.
\newblock {\em Scientific American}, 223(4):120--123, 1970.

\bibitem{ghaffari2017complexity}
Mohsen Ghaffari, Fabian Kuhn, and Yannic Maus.
\newblock On the complexity of local distributed graph problems.
\newblock In {\em Proc.\ 49th Annual ACM Symposium on Theory of Computing (STOC
  2017)}, 2017.
\newblock To appear.
  \href{http://arxiv.org/abs/1611.02663}{\nolinkurl{arXiv:1611.02663}}.

\bibitem{goos13local-approximation}
Mika G{\"o}{\"o}s, Juho Hirvonen, and Jukka Suomela.
\newblock Lower bounds for local approximation.
\newblock {\em Journal of the ACM}, 60(5):39:1--23, 2013.
\newblock
  \href{http://dx.doi.org/10.1145/2528405}{\nolinkurl{doi:10.1145/2528405}}.
  \href{http://arxiv.org/abs/1201.6675}{\nolinkurl{arXiv:1201.6675}}.

\bibitem{hella15weak-models}
Lauri Hella, Matti J{\"a}rvisalo, Antti Kuusisto, Juhana Laurinharju, Tuomo
  Lempi{\"a}inen, Kerkko Luosto, Jukka Suomela, and Jonni Virtema.
\newblock Weak models of distributed computing, with connections to modal
  logic.
\newblock {\em Distributed Computing}, 28(1):31--53, 2015.
\newblock
  \href{http://dx.doi.org/10.1007/s00446-013-0202-3}{\nolinkurl{doi:10.1007/s00446-013-0202-3}}.
  \href{http://arxiv.org/abs/1205.2051}{\nolinkurl{arXiv:1205.2051}}.

\bibitem{kuhn16local}
Fabian Kuhn, Thomas Moscibroda, and Roger Wattenhofer.
\newblock Local computation: lower and upper bounds.
\newblock {\em Journal of the ACM}, 63(2):17:1--17:44, 2016.
\newblock
  \href{http://dx.doi.org/10.1145/2742012}{\nolinkurl{doi:10.1145/2742012}}.
  \href{http://arxiv.org/abs/1011.5470}{\nolinkurl{arXiv:1011.5470}}.

\bibitem{kuusisto14infinite}
Antti Kuusisto.
\newblock Infinite networks, halting and local algorithms.
\newblock In {\em Proc.\ 5th International Symposium on Games, Automata, Logics
  and Formal Verification (GandALF 2014)}, volume 161 of {\em Electronic
  Proceedings in Theoretical Computer Science}, pages 147--160, 2014.
\newblock
  \href{http://dx.doi.org/10.4204/EPTCS.161.14}{\nolinkurl{doi:10.4204/EPTCS.161.14}}.
  \href{http://arxiv.org/abs/1408.5963}{\nolinkurl{arXiv:1408.5963}}.

\bibitem{linial92locality}
Nathan Linial.
\newblock Locality in distributed graph algorithms.
\newblock {\em SIAM Journal on Computing}, 21(1):193--201, 1992.
\newblock
  \href{http://dx.doi.org/10.1137/0221015}{\nolinkurl{doi:10.1137/0221015}}.

\bibitem{naor95what}
Moni Naor and Larry Stockmeyer.
\newblock What can be computed locally?
\newblock {\em SIAM Journal on Computing}, 24(6):1259--1277, 1995.
\newblock
  \href{http://dx.doi.org/10.1137/S0097539793254571}{\nolinkurl{doi:10.1137/S0097539793254571}}.

\bibitem{neumann1966theory}
John von Neumann.
\newblock {\em Theory of Self-Reproducing Automata}.
\newblock University of Illinois Press, 1966.

\bibitem{wolfram2002new}
Stephen Wolfram.
\newblock {\em A New Kind of Science}.
\newblock Wolfram Media, 2002.

\end{thebibliography}

\end{document}